\documentclass[11pt]{article}

\usepackage{amsfonts}
\usepackage{epsfig}
\usepackage[percent]{overpic}
\usepackage{amsmath}
\numberwithin{equation}{section}

\usepackage{amssymb}
\usepackage{graphics}
\usepackage{verbatim}
\usepackage{psfrag}
\usepackage{subfigure}
\usepackage{color}
\usepackage{amsthm}
\usepackage[all]{xy}
\usepackage{makeidx}
\usepackage{enumerate}
\usepackage[letterpaper,margin=0.9in]{geometry}

\let\oldbibliography\thebibliography
\renewcommand{\thebibliography}[1]{%
  \oldbibliography{#1}%
  \setlength{\itemsep}{-1.2mm}%
}

\theoremstyle{plain}
\newtheorem{thm}{Theorem}[section]
\newtheorem{cor}[thm]{Corollary}
\newtheorem{lem}[thm]{Lemma}

\theoremstyle{definition}

\newtheoremstyle{myremark}
  {3pt}
  {3pt}
  {\small \rmfamily}
  {5pt}
  {\rmfamily}
  {:}
  {.5em}
  {}

\theoremstyle{myremark}

\def\R{\mathbb{R}}

\def\N{\mathbb{N}}

\def\cC{\mathcal{C}}

\def\txtD{{\textnormal{D}}}

\def\ra{\rightarrow}
\def\I{\infty}

\newcommand{\be}{\begin{equation}}
\newcommand{\ee}{\end{equation}}
\newcommand{\benn}{\begin{equation*}}
\newcommand{\eenn}{\end{equation*}}
\newcommand{\bea}{\begin{eqnarray}}
\newcommand{\eea}{\end{eqnarray}}
\newcommand{\beann}{\begin{eqnarray*}}
\newcommand{\eeann}{\end{eqnarray*}}

\newcommand{\myendex}{$\blacklozenge$\end{ex}}
\newcommand{\myendexerc}{$\lozenge$\end{exerc}}
\newcommand{\myendpexerc}{$\lozenge$\end{pexerc}}

\def\XXint#1#2#3{{\setbox0=\hbox{$#1{#2#3}{\int}$}
\vcenter{\hbox{$#2#3$}}\kern-.5\wd0}}

\makeindex

\begin{document}

\author{Christian Kuehn\thanks{
Department of Mathematics, Technical University 
of Munich, 85748 Garching b.~M\"unchen, Germany}}

\title{Mathematical Models for Self-Adaptive Response to Cancer Dynamics}

\maketitle

\begin{abstract}
We consider two minimal mathematical models for cancer dynamics and self-adaptation. We aim to capture the interplay between the rapid progression of cancer growth and the possibility to leverage and enhance self-adaptive defense mechanisms of an organism, e.g., motivated by immunotherapy. Yet, our two models are more abstract and generic encapsulating the essence of competition between rapid cancer growth and the speed of adaptation. First, we propose a four-dimensional ordinary differential equation model on a macroscopic level. The model has three main parameter regimes and its most important features can be studied analytically. At small external stimulation of adaptation, cancer progresses to a deadly level. At intermediate external input, a bifurcation is observed \textcolor{black}{leaving behind only a locally stable cancer-free steady state}. Unfortunately, in a large parameter regime after the bifurcation, a transient cancer density spike is observed. Only for a combination of external input and speeding up adaptation, cancer does not reach critical levels. To study the adaptation speed-up in the initial dynamics phase, we switch to a microscopic probabilistic model, which we study numerically. The microscopic model undergoes a sharp transition under variation of the self-adaptation probability. It is shown that a combination of temporal memory and rare stochastic positive adaptation events is crucial to move the sharp transition point to a desired regime. Despite their extreme simplicity, the two models already provide remarkably deep insight into the challenges one has to overcome to leverage/enhance the self-adaptation of an organism in fighting cancer.   
\end{abstract}

\section{Introduction}
\label{sec:intro}

On the one hand, cancer is a highly multi-faceted disease~\cite{JonesBaylin,LujambioLowe,Uhlenetal}. There are many causes, different forms, and progressions~\cite{BissellHines,IkushimaMiyazono}. On the other hand, all cancers follow similar macroscopic principles, i.e., the main disease mechanism is simply based around uncontrolled cell division~\cite{Weinberg1}. Regarding potential therapies, there is an analogous situation~\cite{SouhamiTobias}. For different cancer types, the currently most successful therapies often require specific treatment~\cite{Sawyers}, or even combinations of different methods~\cite{Webster1}. Yet, trying to contain the basic mechanism of uncontrolled cell division is the common principle among many therapies such as chemotherapy~\cite{ChabnerRoberts}, radiation therapy~\cite{Baskaretal} or immunotherapy~\cite{MellmanCoukosDranoff}. In this work, we are interested in developing simple multiscale mathematical models to understand, whether it would be possible to design more generic therapies for cancer treatment. We shall point out immediately that this approach should be viewed as a basic research question that cannot be practically relevant already at this early stage. 

The main idea is to turn the common view of cancer dynamics completely around. Instead of understanding very precise, sometimes even molecular, details, we simply accept that any organism is a highly complex system and cancer is a systemic disease. It is impossible to model all the details of an organism to a sufficiently detailed systemic level. The key thought is to turn an intractable complex system into an advantage. Indeed, organisms are complex systems with highly adaptive capabilities. Not only can humans adapt their knowledge and skills by learning~\cite{GallistelMatzel} but also the immune system~\cite{FarmerPackardPerelson}, hormone regulation~\cite{DuftyClobertMoeller}, as well as many other functions of a complex organism only work because they are adaptable to a range of possible situations~\cite{Holland}. As a basic assumption, we could then postulate that an untreated cancer progresses too fast and thereby evades our adaptive capabilities. Nevertheless, one may ask, whether it is possible to leverage and enhance adaptive capabilities more efficiently. \textcolor{black}{For example, this route is emphasized in immunotherapy~\cite{CouzinFrankel,FarkonaDiamandisBlasutig,Fukumuraetal,MellmanCoukosDranoff}. In fact, the adaptive capabilities of the immune system have been recognized in the context of cancer dynamics~\cite{GajewskiSchreiberFu,Lawetal}. Furthermore, it is clear that adaptive effects of cells also play a key role in cancer cells being able to evade immunotheraphy~\cite{CouzinFrankel,FarkonaDiamandisBlasutig,Fukumuraetal,MellmanCoukosDranoff}. In the context of immunotherapy, a} quite significant controlled intervention is necessary, e.g., to extract, prepare, and grow T-cells in a laboratory setting. One may hence ask, whether a fully adaptive approach is possible for certain therapy concepts, i.e., to just present the body with a sufficient input that shares partial features with cancer cells but which can already be regulated/cured in an adaptive feedback loop. As a concrete example, consider again the immune system, which certainly can fight already a wide range of diseases. As long as it got presented a diverse enough set of features of cancers within harmless diseases, one could make the assumption that it eventually will learn, how to also fight the cancer itself \textcolor{black}{and also learn how to avoid cancer cell evasion to immune system attacks}. Of course, this \textcolor{black}{assumption} is not restricted to the immune system but should apply to train other regulatory organism components relevant for cancer progression.     

One may wonder, why it could be beneficial to also consider cancer models with less details, even if more detailed knowledge is available. In many scientific disciplines, a reductionist abstract approach has moved our understanding significantly forward. As a biological example, consider reduced neuron models just representing the basic excitability properties of neurons~\cite{ErmentroutTerman}. A similar remark applies to a variety of biological (phase) oscillator models for synchronization or swarming~\cite{Kuramoto,VicsekZafiris}. Epidemiology does often not require sophisticated models but standard SIS- and SIR-models provide a remarkably accurate description of qualitative dynamics of epidemic spreading~\cite{KissMillerSimon}. The list could be continued to many other disciplines. Here we shall suggest the viewpoint that more abstract models of cancer dynamics could be extremely useful as well since the basic principles are more visible, and hence also more easy to understand, or even to exploit. 


In this work, we propose two abstract models to describe cancer dynamics. First, we are interested in the macroscopic view of adaptivity. We are going to consider an ordinary differential equation (ODE) model~\cite{Murray1,MuellerKuttler} of just four variables: (I) an input variable $r$ providing a reservoir of cells that share some feature with cancer cells but are essentially harmless, (II) a variable $i$ of body components (e.g.~immune cells) that can \textcolor{black}{remove the input $r$ to progress via adaptation to more efficient components}, (III) a variable $s$ of components that arise due to adaptation of $i$ but so that cancer cells can be attacked (e.g.~``super-immune'' cells), and (IV) the density/proportion of cancer cells $c$. The main results from the analysis of this macroscopic density-based model are:

\begin{itemize}
 \item[(R1)] Due to the large cell division rate inherent in cancer cells, the first commonly observed behavior is cancer progression and ensuing death if the density of cancer cells reaches a certain threshold. For large parameter regimes, there is even a steady state with a large/deadly proportion of cancer cells in the model. \textcolor{black}{The vicinity of this steady state gets reached from large regions of initial data.} As such, the model is just reproducing our basic knowledge about macroscopic cancer dynamics.
 \item[(R2)] There is a bifurcation to a regime, where \textcolor{black}{there is no steady state with a non-zero cancer cell concentration} and a smaller initial spike of cancer cells occurs but after a longer time, the cancer calls are removed by the $s$ components. One main control parameter that can regulate this bifurcation is the rate $\gamma$ at which $r$ is generated. 
 \item[(R3)] Yet, the initial cancer cell spike still has exponential growth before it is removed from the system, which indicates the anticipated time scale to react is relatively small. It turns out that a crucial ingredient for this speed is to also increase the parameter $\beta$ to counter-act exponential growth as $\beta$ controls the conversion rate to generate $s$ components.  
\end{itemize}

So far, the ODE results confirm the main problem cancer poses, namely the adaptivity time scale of the organism is often too slow to cope with the progression of the cancer itself. Finding an abstract model to increase the parameters $\beta$ and $\gamma$ in the ODEs is now the second natural task. As before, we want to stress complex adaptive processes. In contrast to our ODE approach, designing a density-based mean-field differential equations model is likely not optimal. In particular, we need to understand and leverage the fact that generating $s$ components will initially occur in very small numbers. This makes discrete stochastic models more appealing to describe the microscopic/initial adaptive process. For example, thinking of an epidemic, we already understand very well that ODE models describe the larger time scales but initially branching/tree-like processes are more accurate and still conceptually very easy~\cite{Allen}. In our context, we aim to understand, how to improve adaptation by presenting $i$ components on a microscopic level a wide diversity of inputs at different times. In particular, we want to check, whether stochasticity and memory effects can help us to achieve the exponential compensation required for our first model. 

The stochastic principle we employ is that randomness may sometimes lead to favorable outcomes, i.e., adaptation via positive random events~\cite{BarrettSchluter,VanLaarhovenAarts}. In particular, we make the assumption that the organism exploits randomness in its adaptive/control process. Furthermore, it is clear that adaptivity also could be enhanced by memory, i.e., keeping variants of new components around that are further refined in future time steps. We encode this dynamical process using a certain fixed number of components each with a fixed number of attributes. At the initial time, all attributes are zero. Cancer cells can only be attacked if all attributes are eventually equal to one for a given component, which we can call a super-component (analogous to transitioning from $i$ to $s$ in our previous ODE model). There is a process adding information (PAI) to turn attributes from zero level into unit/one level (similar to the input variable $r$ in our first model) and there are rare stochastic events turning attributes into ones. We can abstractly think of this process as positive random events (PREs), which could occur effectively in all parts of an organism. Furthermore, there is a certain memory loss that components revert back to all zeros with a certain probability at each time step. The main results from our discrete (in time and space) microscopic model are:

\begin{itemize}
 \item[(R4)] For no PREs, the model requires an extremely effective/large PAI at each time step or a very long time to successfully combat cancer cells and turn all attributes to ones. This abstractly mirrors the situation we know that very high-intervention or long-duration therapies are currently often our only way to aim to combat cancers.
 \item[(R5)] Increasing the probability of PRE events does not seem to make a difference at first. Over substantial time scales, the number of perfect components stays extremely small. Yet, eventually a very drastic transition happens and the number perfect components becomes very high within a small period of time.
 \item[(R6)] Furthermore, a combination of PAI and PRE seems to be the most effective route to trigger the transition as early as possible in time.
\end{itemize}

The last two conclusions (R5)-(R6) are particularly interesting for potential cancer treatment using self-adapting principles. The challenge is that successful self-adapting transitions are, due to their nature as a sudden transition, very difficult to measure during treatment. This requires extremely careful planning and considerable testing to ensure that a longer-term therapy can succeed as progress monitoring could be difficult. Potentially this might be an explanation, why we have struggled to develop cancer therapies as it might be difficult to measure treatment success for ideas that are at least partially based upon utilizing self-adaptation. In some sense, this could also explain, why there are currently many proposals to consider vaccination-based ideas~\cite{Saxenaetal}. They increase the effective treatment time scale and could reduce the level of uncertainty, even for other future treatments beyond initial vaccination. 

In summary, the results (R1)-(R6) show that thinking about cancer dynamics with more abstract models can be very insightful. This observation holds true for macroscopic as well as for microscopic space-time scales. Our abstract mathematical results led to the conjecture that using the self-adaptive capacities of an organism could potentially be very useful in the quest to find cancer therapies. Furthermore, we have identified the two major roadblocks to this strategy: (1) the involved time-scales have to be manipulated in an optimal way to achieve a counter-exponential growth effect, and (2) progress monitoring could be extremely difficult due to the microscopic transition effect.\medskip  

The remaining part of the paper is structured as follows: In Section~\ref{sec:ODE}, we define and study the macroscopic ODE model for self-adaptation. We analyze the model using standard techniques from nonlinear dynamics to identify the main bifurcation point and to identify the relevant transient spike dynamics before reaching steady state. In Section~\ref{sec:prob}, we define and analyze the microscopic probabilistic model to understand, how we could increase the probability for self-adaptation on a macroscopic level. The model is analyzed using direct numerical simulations. Finally, in Section~\ref{sec:summary}, we briefly put our findings into a broader context leading to an outlook regarding possible future work.

\section{A Macroscopic Generic Model for Adaptivity in Cancer Dynamics}
\label{sec:ODE}

As outlined in Section~\ref{sec:intro}, we are first interested in an abstract macroscopic model for cancer dynamics and adaptive organism response. We introduce four variables modeling non-negative densities, which all depend upon a time variable $t\in [0,\I)$. First, suppose there is a ``reservoir'' variable $r=r(t)\in \R_{\geq 0}$, which is viewed as the only variable that can be driven/changed externally. It models those components that can be eliminated by a systemic organism response (e.g.~via the immune system) but which contain similarities (e.g.~mechanical, chemical, genetic, etc) to cancer cells. The reservoir gets replenished at a rate $\gamma\geq 0$ and there is an interaction function $F(r,i)$ modeling the interaction to the adaptive body component $i$ with rate $\beta\geq 0$\textcolor{black}{; this interaction is purely modelling the change of type between two components $i$ and $s$ described below}. Second, we consider the adaptive component $i=i(t)\in \R_{\geq 0}$ that we could, e.g., think of as ``immune cells''. This component has intrinsic dynamics $g(i)$ modeling growth/proliferation at rate $\alpha_i> 0$ and carrying capacity $N_i>0$. Furthermore, $i$ components can migrate at rate $\beta$ to a third component $s$ if $r$ interacts with $i$ modeled by a reaction term $G(r,i)$. Third, we consider a component $s=s(t) \in \R_{\geq 0}$ that models an improved or ``super'' version of $i$. There is again intrinsic dynamics $h(s)$ modeling growth/proliferation at rate $\alpha_s> 0$ and carrying capacity $N_s>0$. The super-version of $i$ can also be generated from the interaction term $F(r,i)$ between $r$ and $i$ at a rate $\beta \geq 0$. Fourth, we have the cancer cells $c=c(t) \in \R_{\geq 0}$ with intrinsic dynamics $j(c)$ with growth/proliferation at rate $\alpha_c> 0$ and carrying capacity $N_c>0$. The cancer cells can only be removed in an interaction $J(s,c)$, i.e., by the adaptively formed component $s$. In summary, the corresponding ODE model can be formulated as 
\bea
r'&=& \gamma - \beta F(r,i), \label{eq:riscgen1}\\ 
i'&=& g_i(i;\alpha_i,N_i) - \beta G(r,i),\label{eq:riscgen2}\\ 
s'&=& \beta H(r,i) + h(s;\alpha_s,N_s),\label{eq:riscgen3}\\ 
c'&=& j(c;\alpha_c,N_c) - J(c,s),\label{eq:riscgen4} 
\eea
where we used a semi-column in the functions to emphasize their parameter dependencies. We also refer to the class of ODEs as \textit{risc}-models due to the component names. \textcolor{black}{The structural mathematical form of the \textit{risc}-model is quite similar to other ODE models in biology, e.g., predator-prey dynamics, epidemic spreading, or neural spiking~\cite{Murray1,Murray2,MuellerKuttler}. A key feature is evidently the mechanism to transition from the $i$ components to the $s$ components, which will be studied and modeled in more microscopic detail in Section~\ref{sec:prob}.} Unfortunately, in the form~\eqref{eq:riscgen1}-\eqref{eq:riscgen4} the ODEs are too difficult to analyze as there are many unknown functions in the vector field. For certain simple invariant sets of the problem, a dynamical systems analysis might still be feasible for quite general intrinsic and interaction functions~\cite{KuehnSiegmundGross,KuehnGross,Vassena}. Here we shall not consider these analysis techniques. To explore the model, we just specify the functions by the simplest possible, and still biologically reasonable, forms. For the intrinsic dynamics of $i,s,c$, we will assume logistic growth/proliferation $x\mapsto \alpha_x x(1-x/N_x)$ for $x\in\{i,s,c\}$. Furthermore, for the interaction functions between components we take the simplest predator-prey-type or SI-spreading-type dynamics $(x,y)\mapsto xy$ that we know from basic biological modeling. With these choices, our ODE model reduces to the following form 
\bea
r'&=& \gamma - \beta ri,\label{eq:risc1}\\ 
i'&=& \alpha_i i\left(1-\frac{i}{N_i}\right) - \beta ri,\label{eq:risc2}\\ 
s'&=& \beta ri + \alpha_s s\left(1-\frac{s}{N_s}\right),\label{eq:risc3}\\ 
c'&=& \alpha_c c\left(1-\frac{c}{N_c}\right) - sc.\label{eq:risc4} 
\eea
As~\eqref{eq:risc1}-\eqref{eq:risc4} is a 'basic' variant of the \textit{risc} model class, we also refer to it as the \textit{brisc} model. \textcolor{black}{Of course, there is a significant simplification step for any low-dimensional ODE model in the context of biology. No simplified ODE is a completely biologically accurate representation of the actual processes. Yet, these models turned out to be very useful for our conceptual understanding. Next, from a mathematical viewpoint local} existence of unique smooth ODE solutions of the \textit{brisc} model is clear due to local Lipschitz continuity and smoothness of the vector field. The mathematical question of global-in-time existence is not of primary relevance for us here as we are only interested in the initial dynamics as well as potential well-defined long-term bounded solutions. If any one of the components have a finite-time blowup, then the resulting solutions do not really have an immediate biological interpretation in our context. Figure~\ref{fig:01} shows a direct numerical simulation of the \textit{brisc} model\footnote{Simulations were carried out using a robust stiff variable-order numerical method (\texttt{ode15s}) in MatLab 2024a.}. We observe that for no self-adaptation ($\beta=0$), the cancer very quickly reaches a comparable density to $i$ components, which would certainly mean death of the organism. For small self-adaptation ($\beta=0.1$), we see that eventually a decrease in cancer cells is reached but the initial spike is again so large that again we expect the cancer to be deadly. Only for significantly larger values ($\beta=20$), we notice that the initial spike of cancer density is lowered while it eventually even dis-appears on a long time scale. Basically, one could \textcolor{black}{conjecture that} these three cases to the three common cancer progressions: (I) relatively rapid death, (II) decrease in cancer cells but possible persistence/re-emergence of cancer cells, and (III) complete cancer removal (which eventually has to take place via stochastic events, which we do not model here at small density of $c$). \textcolor{black}{In this regard, we remark that this type of trichotomy is actually common in mathematical models for disease dynamics, e.g., in well-known SIRS type population models.} 

\begin{figure}[htbp]
	\centering
	\begin{overpic}[width=1.0\textwidth]{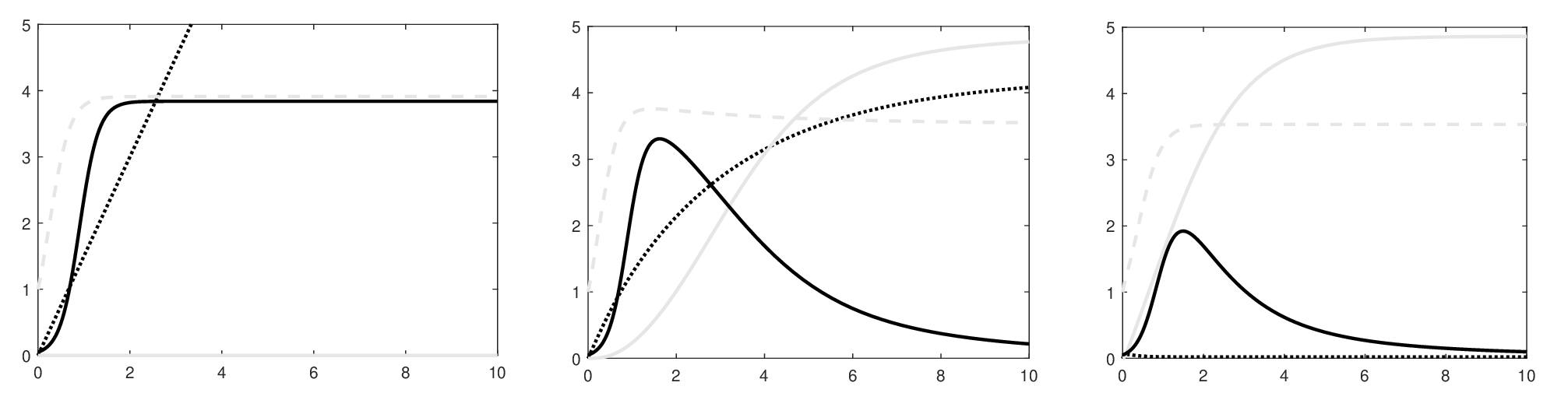}
		\put(17,0){$t$}
		\put(51,0){$t$}
		\put(84,0){$t$}
		\put(13,25){$\beta=0$}
		\put(48,25){$\beta=0.1$}
		\put(80,25){$\beta=20$}
	\end{overpic}
	\caption{\label{fig:01}Numerical integration of the \textit{brisc} model~\eqref{eq:risc1}-\eqref{eq:risc4}. The parameter $\beta\in\{0,0.1,20\}$ is varied \textcolor{black}{for the same initial condition $(r(0),i(0),s(0),c(0))=(0,1,0,0.05)$} and fixed parameters $\gamma=1.5$, $\alpha_i=4.4$, $\alpha_s=0.3$, $\alpha_c=4.8$ and $N_i=N_s=N_c=4$. The variables of the times series of the four variables $r$ (black dotted), $i$ (grey dashed), $s$ (grey solid) and $c$ (black solid), are plotted. We observe that the strength of the self-adaptation $\beta$ can indeed be crucial for the progression of a cancer\textcolor{black}{, i.e., for higher self-adaptation the peak concentration is lowered}.}
\end{figure}
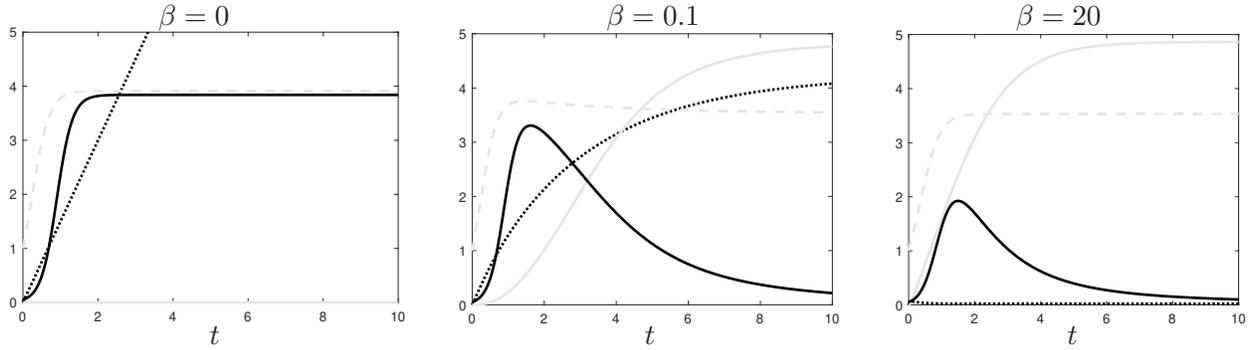

We note that we do not normalize the components $i,s,c$. This is not needed as we are only interested in the relative growth over time and whether $c$ can reach a given threshold value, say $c\geq c^*>0$. This threshold marks the point, where recovery of the organism is impossible and effectively immediate death occurs. \textcolor{black}{Setting the threshold $c^*$ is somewhat open as it might be specific for each organism with cancer, so we stress that multiple choices are possible but we eventually have to select one option}. For example, it could be reasonable to assume that such a $c^*$ must be lower than the steady state of the component $i^*$ that can be present in the body before cancer onset at steady state\textcolor{black}{, which is the choice we make here. Just for a broader perspective, we mention that an alternative choice could be to select the threshold $c_*$ relative to the carrying capacity $N_c$ but this then raises the issue how far below this threshold we want to stay.} The value $i^*$ can be defined by studying~\eqref{eq:risc2} with $\beta=0$, which yields
\be
\label{eq:ialone}
i'= \alpha_i i\left(1-\frac{i}{N_i}\right).
\ee 
There are two steady states for~\eqref{eq:ialone} given by the trivial steady state $i=0$ and the nontrivial steady state $i=N_i$. By direct linearization, we check that $i=0$ is unstable as we assume $\alpha_i>0$. The nontrivial steady state is globally stable in the biologically relevant region for \eqref{eq:ialone} as long as we consider $i(0)>0$, which can be seen by a trivial argument as one-dimensional ODEs are gradient systems. So if we would set, e.g., $i^*=N_i/2$, then we see a clear difference in behavior in Figure~\ref{fig:01} between the cases $\beta=0.1$ and $\beta=20$. Indeed, for $\beta=0.1$ we cross $i^*$ in time, while we slightly stay below it for $\beta=20$. Before going into more detailed transient dynamics analysis, let us first work out the basic time-asymptotic properties for the \textit{brisc} model~\eqref{eq:risc1}-\eqref{eq:risc4}. 

\begin{lem}
\label{lem:steady}
There are at most twelve steady states $(r_*,i_*,s_*,c_*)$ of~\eqref{eq:risc1}-\eqref{eq:risc4} given by
\benn
(r_*,i^\pm_*,s^\pm_*,c_*)=\left(\frac{\gamma}{\beta i_*},\frac{N_i}{2}\pm\sqrt{\frac{N_i^2}{4}-\frac{\gamma N_i}{\alpha_i}},\frac{N_s}{2}\pm \sqrt{\frac{N_s^2}{4}+\frac{\gamma N_s}{\alpha_s}},c_*\right),
\eenn 
where there are three possible options for $c_*$ given by $c_*=0$ or $c_*^\pm=N_c(1-s_*^\pm/\alpha_c)$ (where $s_*^\pm$ could take two values as stated above).
\end{lem}

\begin{proof}
If we solve for the steady states of~\eqref{eq:risc2}, then one possibility is $i_*=0$. However, this is incompatible with \eqref{eq:risc1} as we assumed $\gamma>0$. So we can rule out $i_*=0$. Then~\eqref{eq:risc1} gives $r_*=\gamma/(\beta i_*)$ and hence~\eqref{eq:risc2} yields
\benn
0\stackrel{!}{=}\alpha_i(1-i_*/N_i)-\beta r_*=\alpha_i(1-i_*/N_i)-\gamma/ i_* \quad \Rightarrow~i_*^\pm=\frac{N_i}{2}\pm\sqrt{\frac{N_i^2}{4}-\frac{N_i\gamma}{\alpha_i}}.
\eenn 
From~\eqref{eq:risc3}-\eqref{eq:risc4} we can then obtain with a similar direct calculation the values for $s_*^\pm$ and $c_*$. Then we just count that $r_*$ is fixed, then $i_*^\pm$ gives two solutions, which leads then to maximally four solutions from the formula for $s_*^\pm$. Finally, since there are three possible solutions for $c_*$, this gives at most twelve possible solutions.   
\end{proof}

We do not expect that all steady states that are algebraically possible will matter for the biological application. However, one should be careful to not eliminate too many phenomena by making too stringent assumptions. Let us illustrate this last point by a certain parameter regime. It might be tempting to biologically assume that $\alpha_c$ is very large as cancers grow very quickly and set $\varepsilon:=1/\alpha_c$ to obtain a fast-slow ODEs~\cite{KuehnBook,Hek} 
\bea
r'&=& \gamma - \beta ri,\label{eq:risc1fs}\\ 
i'&=& \alpha_i i\left(1-\frac{i}{N_i}\right) - \beta ri,\label{eq:risc2fs}\\ 
s'&=& \beta ri + \alpha_s s\left(1-\frac{s}{N_s}\right),\label{eq:risc3fs}\\ 
c'&=& \frac{1}{\varepsilon} c\left(1-\frac{c}{N_c}\right) - sc.\label{eq:risc4fs} 
\eea
Unfortunately, the next result illustrates, why the limit $\varepsilon\ra 0$ is not immediately suitable for the problem.

\begin{lem}
The critical manifold $\cC_0$ of \eqref{eq:risc1fs}-\eqref{eq:risc4fs} is given by two components
\benn
\cC_0=\{(r,i,s,c)\in\R^4_{\geq 0}:c=0\}\cup \{(r,i,s,c)\in\R^4_{\geq 0}:c=N_c\}=:\cC_{\textnormal{r}} \cup \cC_{\textnormal{a}}
\eenn
where $\cC_{\textnormal{r}}$ is normally hyperbolic repelling and $\cC_{\textnormal{a}}$ is (actually globally) normally hyperbolic attracting.
\end{lem}

\begin{proof}
One just calculates 
\benn
\txtD_c\left(c(1-c/N_c)\right)=1-2c/N_c
\eenn 
and evaluates this expression on the critical manifold. 
\end{proof}

Standard Fenichel theory then gives us the existence of a slow manifold $\cC_\varepsilon$ for $\varepsilon>0$ sufficiently small~\cite{Fenichel4,KuehnBook}. The slow manifold inherits the stability properties of the normally hyperbolic critical manifold. So as long as the initial value satisfies $c(0)>0$, which is natural to assume for our setting, then we get quickly attracted to $\cC_{\textnormal{a}}$. Unfortunately, this does not cover the transition regime we are interested in, where it should be possible to reach different asymptotic cancer densities. The pure singular limit $\alpha_c\ra \I$ is not of immediate practical relevance. A similar remark applies to other singular limits, e.g., $\gamma,\beta\ra \I$ is practically impossible to reach. Hence, we have to do a full stability/bifurcation analysis.

\begin{lem}
\label{lem:physss}
Restricting the phase space to $(r,i,s,c)\in\R^4_{\geq 0}$, there are most four possible steady states
\benn
(r_*,i_*^\pm,s_*,c_*)=\left(\frac{\gamma}{\beta i_*},\frac{N_i}{2}\pm\sqrt{\frac{N_i^2}{4}-\frac{\gamma N_i}{\alpha_i}},\frac{N_s}{2}+\sqrt{\frac{N_s^2}{4}+\frac{\gamma N_s}{\alpha_s}},c_*\right),
\eenn 
where $c_*=0$ or $c_*=N_c(1-s_*/\alpha_c)$. Furthermore, there is
\begin{itemize}
 \item[(L1)] only one possible non-trivial steady state with $c_*>0$ if 
\be
\label{eq:onlycancerzero}
N_i\geq\frac{4\gamma }{\alpha_i} \qquad \text{and}\qquad \alpha_c\leq \frac{N_s}{2}+\sqrt{\frac{N_s^2}{4}+\frac{\gamma N_s}{\alpha_s}}.
\ee
 \item[(L2)] no possible steady states if 
\be
\label{eq:limitedcarryi}
N_i<\frac{4\gamma }{\alpha_i}.
\ee
\end{itemize}
\end{lem}

\begin{proof}
There is only one possible steady state value $s_*^+=s_*$ in the biological region as we just estimate for $s^-_*$ that
\benn
\frac{N_s}{2}- \sqrt{\frac{N_s^2}{4}+\frac{\gamma N_s}{\alpha_s}}< \frac{N_s}{2}- \sqrt{\frac{N_s^2}{4}}=0. 
\eenn
So there is only one possibility for $s_*$ and hence only two possibilities for $c_*$, which then yields at most four possible steady states. To prove the last two statements, we just look carefully at the arguments of the square-roots in $i_*^\pm$ and check when $c_*$ is negative.  
\end{proof}

We already observe that the condition~\eqref{eq:onlycancerzero} is absolutely crucial as it means that the nonzero cancer density state state can disappear under parameter variation. In particular, we have as a direct corollary of the algebraic form of the steady states: 

\begin{cor}
\label{cor:transcritical}
\textcolor{black}{There exists a fold bifurcation when $N_i=\frac{4\gamma }{\alpha_i}$. Furthermore, the non-trivial cancer equilibrium with $c_*>0$ leaves the physical domain when $\alpha_c= \frac{N_s}{2}+\sqrt{\frac{N_s^2}{4}+\frac{\gamma N_s}{\alpha_s}}$ at a transcritical bifurcation.}
\end{cor}

\textcolor{black}{In particular, we see that increasing $\gamma$ is one natural route to eventually remove all steady states. This is very reasonable from a modeling perspective as $\gamma$ is the parameter that inputs additional $r$ components into the system that produce $s$ components, which eliminate $c$. Furthermore, tuning other system parameters can also remove the non-trivial cancer state with $c_*>0$ from the physical domain.} Although the main bifurcation points are now identified, let us also look at the eigenvalues for completeness:

\begin{lem}
\label{lem:eigenvalues}
The Jacobian at the steady state from Lemma~\ref{lem:steady} is given by
\be
\label{eq:Jacobian}
\left(
\begin{array}{cccc}
-\beta i_* & -\beta r_* & 0 & 0 \\
-\beta i_* & \alpha_i-2\alpha_i i_*/N_i-\beta r_* & 0 & 0\\
\beta i_* & \beta r_* & \alpha_s -2\alpha_s s_* & 0\\
0 & 0 & -c_* & \alpha_c-2\alpha_c c_*/N_c -s_* 
\end{array}
\right).
\ee
The eigenvalues $\lambda_j$ for $j\in\{1,2,3,4\}$ are given by 
\benn
\lambda_1=\alpha_c-s_*-2c_*\alpha_c/N_c,\quad \lambda_2=\alpha_s-2s_*\alpha_s,~\lambda_{3,4}=\lambda_{3,4}(r_*,i_*,N_i,N_c,\beta).
\eenn
\end{lem}

\begin{proof}
The first result \eqref{eq:Jacobian} follows directly by differentiation. The first two eigenvalues follow from lower right block of the Jacobian and the remaining two eigenvalues can be calculated by algebraic formulas directly\footnote{The actual double square-root formulas are not very insightful, we shall just work with these in suitable limits.}. 
\end{proof}

From the eigenvalue $\lambda_1$ we find, as expected from Corollary~\ref{cor:transcritical}, the bifurcation condition that the eigenvalue vanishes along the subspace $\{c_*=0\}$ precisely when $\alpha_c=s_*$. In principle, one could now also calculate a full one-dimensional center manifold reduction and find the normal form coefficients for the bifurcation on the center manifold. Yet, this will not lead to any additional biological insight in comparison to the direct algebraic formulas that gave Corollary~\ref{cor:transcritical}. 

As such, one might believe that this solves the problem, just increase $\gamma$ \textcolor{black}{or change other system parameters suitably}. The situation is not so simple unfortunately, e.g., we know that $N_i\alpha_i-4\gamma<0$ yields no steady states at all, so arbitrary increase in $\gamma$ \textcolor{black}{may} not the solution as this would deplete $i$ too quickly. One might \textcolor{black}{potentially} interpret this biologically as the simplest manifestation of a ``cytokine storm'' for general self-adaptation, where too much self-adaptation damages the basics of an organism; it might be reasonable to call this effect self-adaptation storm as it may in principle happen also in many other regulatory/self-adapting systems, not only the immune system. Yet, in practical terms it might be possible to control $\gamma$ to avoid it getting too large, e.g., by slowly increasing $\gamma$. \textcolor{black}{In addition to analytical calculations, we have also performed numerical continuation to cross-validate our results and to check, whether some interesting higher co-dimension bifurcations are easy to find; see Figure~\ref{fig:06}. We observe that for $\gamma$ sufficiently large, all steady states dis-appear as expected via a fold bifurcation but we could not locate easily any interesting higher co-dimension bifurcations except for a very degenerate point at the boundary of the parameter domain.}

\begin{figure}[htbp]
	\centering
	\begin{overpic}[width=1.0\textwidth]{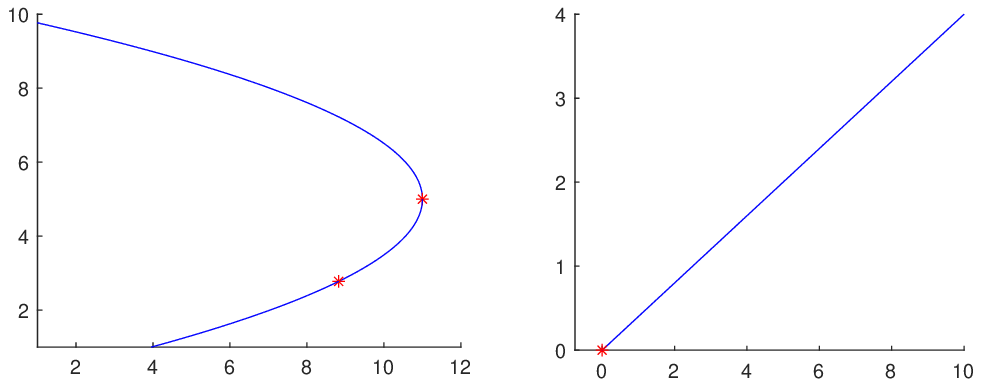}
		\put(0,33){$i_*$}
		\put(43,0){$\gamma$}
		\put(94,0){$\gamma$}
		\put(54,33){$\alpha_i$}
		\put(37,35){(a)}
		\put(80,35){(b)}
		\put(43,16){LP}
		\put(34,7.5){H}
		\put(59,4){CP}
	\end{overpic}
	\caption{\textcolor{black}{\label{fig:06}Numerical continuation carried out in MatCont~\cite{Dhoogeetal} for fixed parameters $\beta=1$, $N_i=10$, $\alpha_s=0.3$, $N_s=10$, $\alpha_c=4.8$, and $N_c=10$. (a) Continuation of a family of equilibria with $c_*=0$ in the main bifurcation parameter $\gamma$. We see that there exists a fold bifurcation (LP, limit point) so that if $\gamma$ is sufficiently large, two steady states disappear as predicted analytically. Also, a Hopf point (H) is detected, which is just a neutral saddle and does not seem to play a major role in the dynamics. (b) Two-parameter continuation in $(\gamma,\alpha_i)$ parameter space for the fold bifurcation point. At the degenerate origin in parameter space a cusp point (CP) is detected. Yet, over the physically relevant wide ranges of positive parameters there is no other co-dimension two bifurcation.}}
\end{figure}
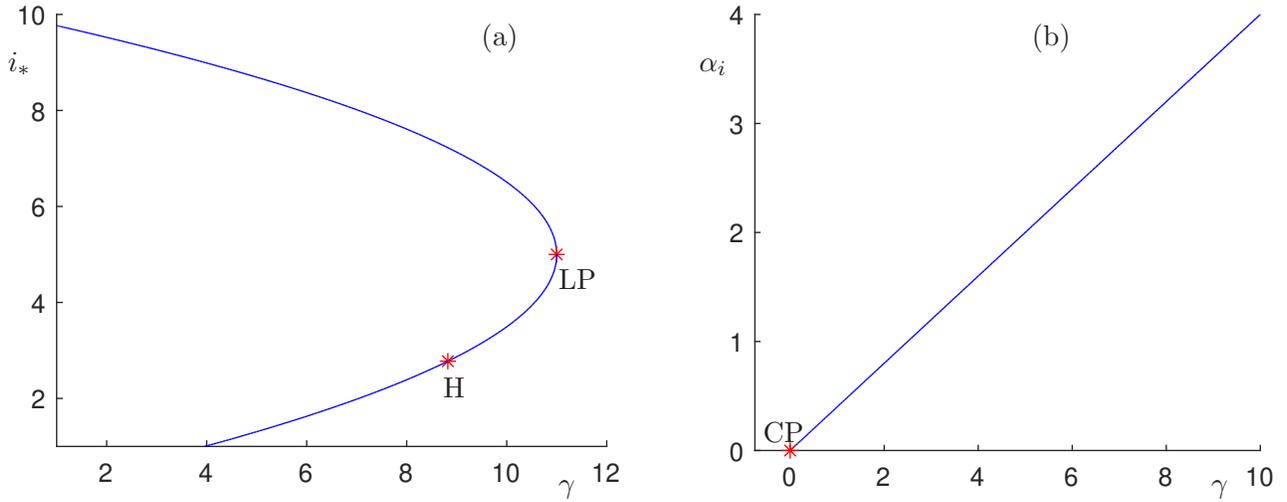 

Unfortunately, even if we achieve the death of $s$ components with a moderate $\gamma$, utilizing self-adaptivity is not that simple. We have already seen in Figure~\ref{fig:01} for $\beta=0.1$ that the rate $\beta$ at which the organism can self-adapt and produce $s$ components might still be too low, even though the nontrivial cancer steady state for $c_*>0$ is not reached. The transient behavior is potentially still too dangerous as the spike in $c$ is potentially too large as discussed already above. A first hint, why the parameter $\beta$ also plays a key role is most easily visible if we look at the eigenvalues:

\begin{lem}
$\lim_{\beta\ra \I} \lambda_{3,4}\in\{0,-\I\}$.
\end{lem} 

\begin{proof}
This is immediate from the upper left two-by-two block of the Jacobian.
\end{proof}

Therefore, at least the local attraction speed to the cancer-free equilibrium $c_*=0$ depends crucially on the size of $\beta$. A more global convergence analysis and/or bifurcation analysis could be carried out but we shall postpone this to future work as our main goal here is just to focus on the key conclusions from the \textit{brisc} model. In summary, we see that the actual input of $r$ components via increasing $\gamma$ makes it principally possible to use adaptive capabilities to remove cancer cells $c$. We have seen that the conversion rate $\beta$ is crucial in the step of convergence to the steady state and large transient peaks can appear if $\beta$ is too small. Therefore, we have think about a microscopic model of adaptation to understand, how one might be able to increase the effectiveness of self-adapting mechanisms. Increasing $\gamma$ is externally controlled, so the best we can hope for is that eventually we will be able to produce sufficiently many $r$ components. Yet, we can try to understand better, how we might be able to increase the conversion rate $\beta$.

\section{A Microscopic Generic Model for Adaptivity in Cancer Dynamics}
\label{sec:prob}

From Section~\ref{sec:ODE}, we have already learned that it would be incredibly useful to understand the transformation between $i$ components and $s$ components via a discrete-time and discrete-space model. \textcolor{black}{From a biological viewpoint, one motivation is to consider the immune system and its competition with cancer evolution~\cite{TholPawlikMcGranahan}. It is well-understood that cancer cells have to be able to evade the immune response to facilitate tumor growth. Yet, there are biologically well-motivated conjectures~\cite{Baxevanisetal} that improved immunotheraphy can help us to target cancer cells. On an abstract mathematical level, this is precisely an example of a process improving the targeted abilities of $i$ components to $s$ components. The reason, why we have decided to use an added layer of abstraction is that the body response to cancer dynamics is definitely not limited to just immune cells. For example, the hormone system may also play a crucial role in many cancers~\cite{Modugnoetal,RastelliCrispino}. Its regulating cells could equally well be targets for adaptation, i.e., promoting them from $i$ components to $s$ components. Very similar remarks also apply to recent discoveries regarding neurological control systems in cancer progression~\cite{Khanmammadovaetal}. Since we are still far away from a detailed understanding and interplay of all these biological systems and their precise role in cancer dynamics, it seems worthwhile to abstract the main possibility, i.e., that there are potential switches that can help the body to respond to a cancer.} We just want to capture the initial phase when $s$ components can be created, so we fix a finite final time $T\in \N$. Furthermore, we consider as a phase space of the dynamical system the space of matrices $A=(a_{jk})\in\R^{M\times N}$ with entries $a_{jk}\in \{0,1\}$ for all $j,k\in\N$. Every column abstractly represent one $i$ component that can potentially transform into an $s$ component. The $M$ rows of the matrix $A$ represent different attributes where an entry $1$ means we have satisfied/obtained the attribute. A full vector of ones, i.e., if $(a_{jk})_{j=1}^N=(1,1,\ldots,1)^\top$, means that the component $i$ has been transformed into a component $s$. The dynamics is then defined as follows:

\begin{itemize}
 \item[(S0)] We fix $M,N\in\N$, and $T\in\N$ with $\N=\{1,2,\ldots\}$. At time $t=0$, we initialize $A=0\in\R^{M\times N}$, i.e., we start with the zero matrix consisting only of ``pure'' $i$ components.
 \item[(S1)] At each time step $t\in\N$ with $t<T$, we carry out three possible dynamical transitions: 
 \item[(S1a)] we turn a uniformly \textcolor{black}{at random chosen row} of $A$ into all ones;
 \item[(S1b)] with probability $p_d\in(0,1)$, we turn a uniformly at random chosen column into zeros;
 \item[(S1c)] with probability $p_m\in(0,1)$, each zero entry can turn into a one entry;
 \item[(S2)] we increase the time from $t$ to $t+1$ and return to (S1).
\end{itemize} 

The dynamical rules are relatively easy to interpret. The rule (S1a) abstracts our idea of the reservoir variable from Section~\ref{sec:ODE} since turning an entire attribute for all cells into ones means that the $r$ components have triggered an improved situation to convert from $i$ to $s$. \textcolor{black}{An alternative and equivalent description of (S1a) would be that with a probability of $1/M$ a row is selected and all its entries are set to $1$.} We also refer to (S1a) as a process adding information (PAI). Rule (S1b) just models cell/component death as well as the loss of memory. Note carefully that the columns of $A$ essentially have a memory function collecting all the possible triggers we have given the self-adapting system that present features that are similar to the features of the cancer cells. Rule (S1c) accounts for positive random events (PREs) as current zeros might be converted to ones at a very low probability, i.e., we want to pick $p_m$ very small. This models the fact that self-adapting complex systems must have some function to try out different new/creative solutions to a presented problem and some must have positive effects, which are the ones that are kept as self-adapting systems also must have the ability to test such ``mutation'' trials. Note carefully that the time scale might still be very fast for this process as here PREs could just mean tuning the parameters of an existing chemical or hormonal subsystem of the organism to achieve a positive effect. \textcolor{black}{We also note that (S1a) and (S1c) both have positive effects but model very different biological processes, i.e., an external deterministic one in comparison to a rare probabilistic intrinsic one.} 

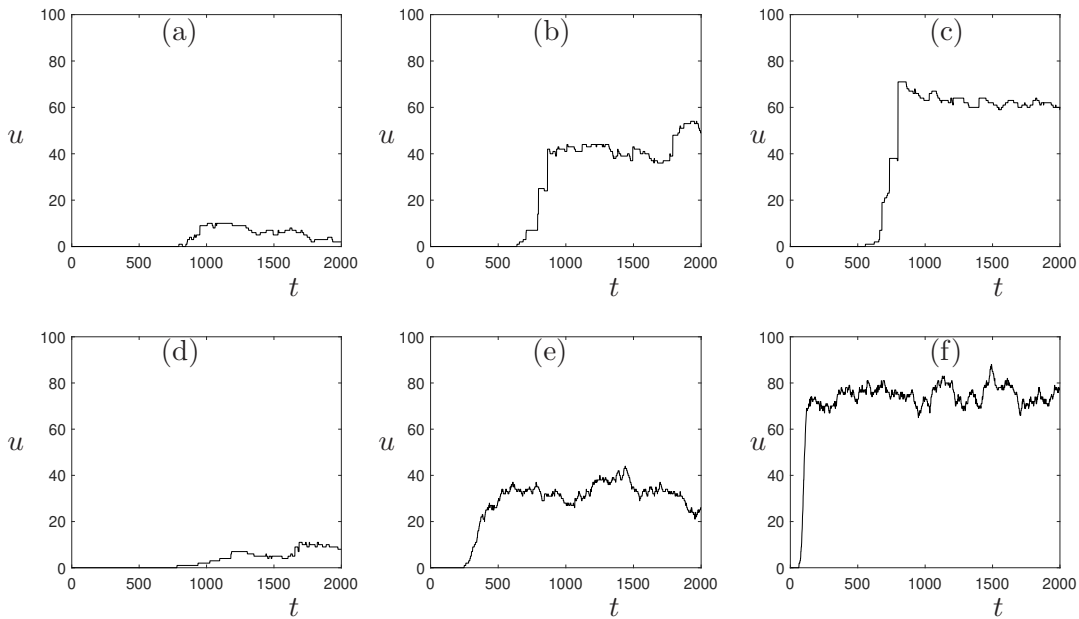
\begin{figure}[htbp]
	\centering
	\begin{overpic}[width=1.0\textwidth]{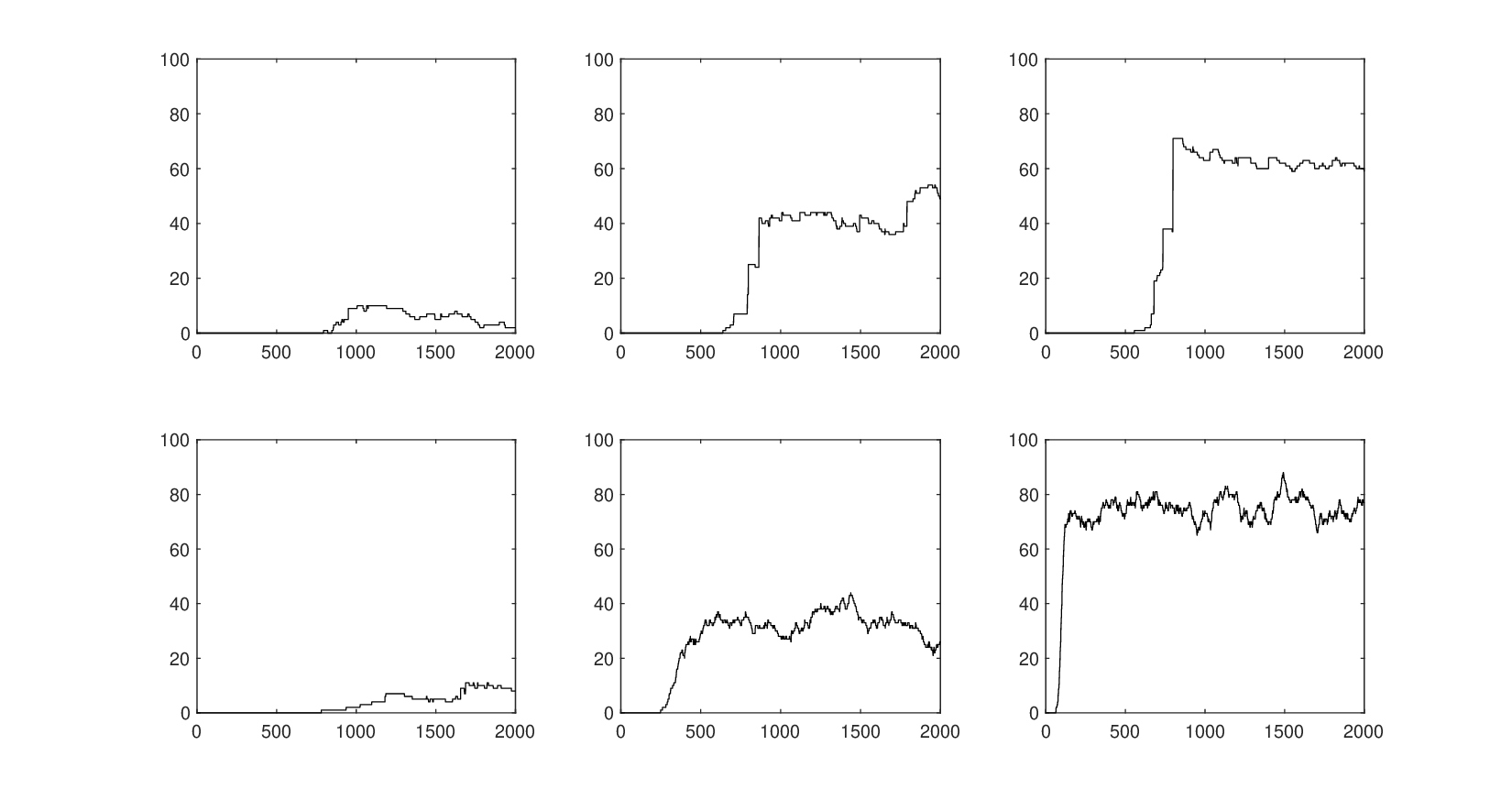}
		\put(8,15){$u$}
		\put(37,15){$u$}
		\put(66,15){$u$}
		\put(8,39){$u$}
		\put(37,39){$u$}
		\put(66,39){$u$}
		\put(30,2){$t$}
		\put(55,2){$t$}
		\put(85,2){$t$}
		\put(30,27){$t$}
		\put(55,27){$t$}
		\put(85,27){$t$}
		\put(20,47){(a)}
		\put(49,47){(b)}
		\put(80,47){(c)}
		\put(20,22){(d)}
		\put(49,22){(e)}
		\put(80,22){(f)}
	\end{overpic}
	\caption{\label{fig:02}Simulation of sample paths \textcolor{black}{for} the model (S0)-(S2) with $M=200$, $N=100$, $T=2000$. (a)-(c) were computed with fixed $p_m=0.001$ and for $p_d=0.3,0.1, 0.05$. The parts (d)-(f) were computed with fixed $p_d=0.3$ and $p_m=0.001,0.01,0.05$. On the vertical axis we show $u$ as defined in~\eqref{eq:micobs}. \textcolor{black}{We observe the emergence of a transition as the probabilities $p_d=0.3$ and $p_m$ are increased.}}
\end{figure}

With these three rules, we may now simulate the model and ask, which key effects arise that could help us to understand how to increase the parameter $\beta$ from Section~\ref{sec:ODE}? Indeed, the size of $\beta$ can be seen as an aggregate outcome of the microscopic model we have just defined as more columns with all ones in the matrix $A$ at time $T$ would just correspond to a higher $\beta$. Figure~\ref{fig:02} shows a direct simulation of our microscopic model, where we selected as an observable 
\be
\label{eq:micobs}
\textcolor{black}{u=u(t):=|\{k:(a_{jk})_{j=1}^M=(1,1,\ldots,1)^\top\}|,}
\ee 
which just counts the current number of \textcolor{black}{columns} at time $t$ that are all ones, i.e., have been turned from $i$ into $s$. As one would expect from the biological interpretation of the model, the value of $u$ is generally higher, when the death rate $p_d$ decreases and the positive mutation rate $p_m$ increases. Yet, Figure~\ref{fig:02} also shows that even if $u$ reaches a high value, which means that most columns have been positively transformed, then this transition occurs rather abruptly from low values to high values of $u$; see Figure~\ref{fig:02}(b)-(c) and (e)-(f). Clearly, such a transition reminds us of classical \textcolor{black}{bifurcation-type or (phase)} transitions in other microscopic probabilistic models~\cite{Kesten,Kuramoto,Liggett}. Generally it is not easy to prove the existence of such transitions analytically, so we shall not pursue an analytical study of the model but rather focus on numerical simulations.   

\begin{figure}[htbp]
	\centering
	\begin{overpic}[width=0.6\textwidth]{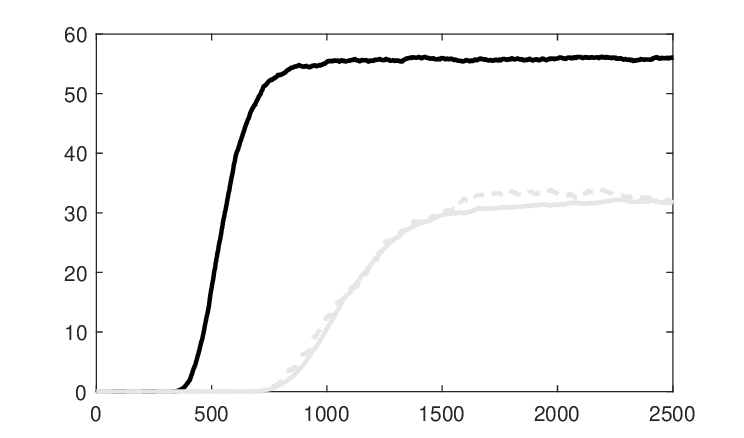}
		\put(4,35){$u$}
		\put(82,0){$t$}
		\end{overpic}
	\caption{\label{fig:03}Simulation averaged over 100 sample paths the model (S0)-(S2) with $M=200$, $N=100$, $T=2500$. The solid black curve was computed for $p_d=0.1$, $p_m=0.005$ and with PAI at each time step turned on. The solid gray curve was computed with the same parameter values but with PAI turned off. The dashed \textcolor{black}{curve} was computed for $p_d=0.1$, $p_m=0$ and with PAI turned on. \textcolor{black}{We observe that the combination of memory and positive mutation working in tandem are key to achieve a quick and early-in-time transition point.}}
\end{figure}

In this regard, a rather remarkable feature of the model is that memory and positive mutation working in tandem are key to achieve a quick and early-in-time transition of $u$ to high values. Figure~\ref{fig:03} illustrates this effect. If we do not have a PAI mechanism that slowly adds ones at each time step in each row and this information is kept for some time, then the transition happens much later and also saturates at a lower value of $u$. Similarly, if we do not have PREs, a very similar outcome is observed in comparison to the combination of both effects. 

\begin{figure}[htbp]
	\centering
	\begin{overpic}[width=0.6\textwidth]{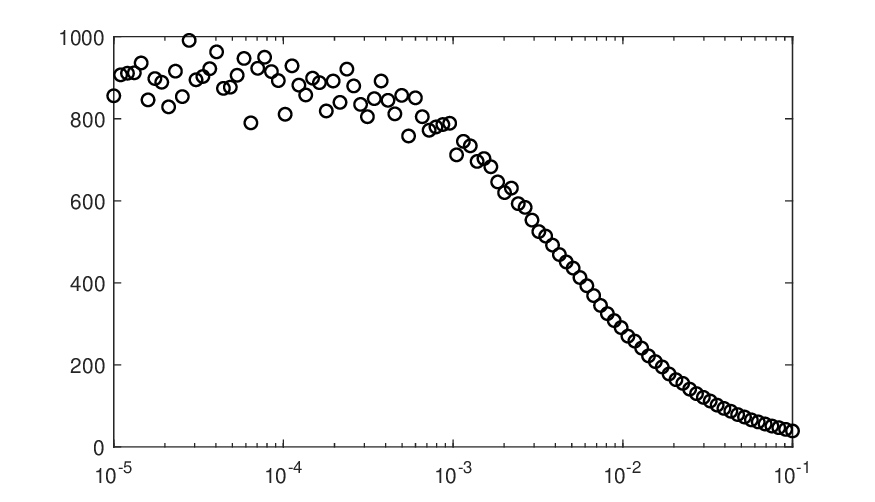}
		\put(3,30){$t_*$}
		\put(80,0){$p_m$}
		\end{overpic}
	\caption{\label{fig:04}Simulation averaged over 30 sample paths per run for the model (S0)-(S2) with $M=200$, $N=100$, $T=2000$ and $p_d=0.1$. We plot the location of the transition point $t_*$ in time depending on the good mutation probability $p_m$. The scale is semi-logarithmic with a log-scale on the horizontal axis.\textcolor{black}{It is interesting to observe that there seems to be a critical threshold above which the good mutation probability $p_m$ starts to induce a significant shortening of the time to the transition.}}
\end{figure}

We can also study the location of the transition point $t_*$ in time as we increase probability $p_m$ at which PREs happen. Figure~\ref{fig:04} shows this dependence, where we computed the transition time/point $t_*$ from direct simulations; effectively the figure shows a relatively tight upper bound to $t_*$ as we determined $t_*$ has occurred once $u$ reaches a level of ten percent of column vectors that are all ones. We see that $t_*$ has a first phase where it does not significantly change for low values of $p_m$. This is quite natural as this just means the PRE effect (S1c) contributes too little to the transition, which is then primarily driven by the PAI effect (S1a). Yet, there is a broad range of values between $10^{-3}$ and $10^{-1}$ visible in Figure~\ref{fig:04}, where the transition does occur much earlier due to the increased mutation rate $p_m$. Finally, the situation eventually saturates as the transition point $t_*$ moves closer to $t=0$ as $p_m$ is increased even further. 

\begin{figure}[htbp]
	\centering
	\begin{overpic}[width=0.6\textwidth]{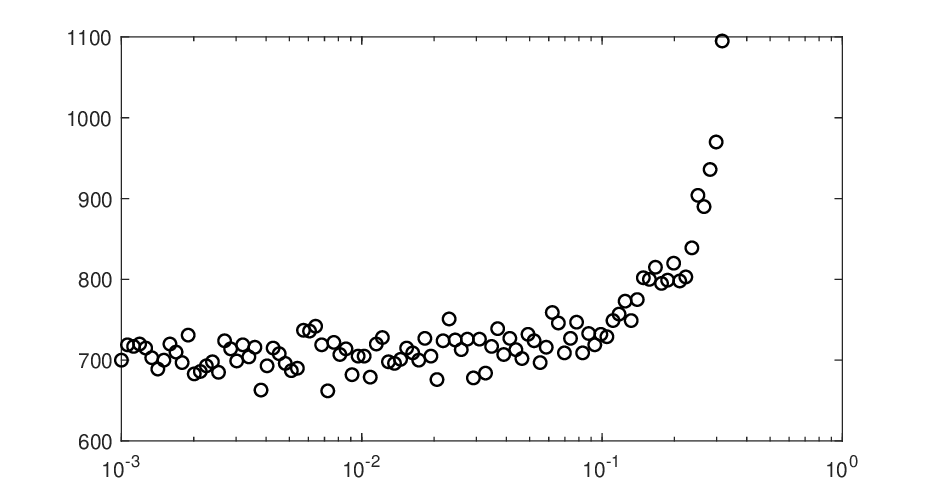}
		\put(3,30){$t_*$}
		\put(80,0){$p_d$}
		\end{overpic}
	\caption{\label{fig:05}Simulation averaged over 30 sample paths per run for the model (S0)-(S2) with $M=200$, $N=100$, $T=2000$ and $p_m=0.001$. We plot the location of the transition point $t_*$ in time depending on the death probability $p_d$. The scale is semi-logarithmic with a log-scale on the horizontal axis. \textcolor{black}{Analogous to Figure~\ref{fig:04}, it is interesting to observe that there seems to be a critical threshold above which the probability $p_d$ starts to induce a significant extension of the time to the transition.}}
\end{figure}

In addition to the probability $p_m$, we can also consider the probability $p_d$ as a parameter. Figure~\ref{fig:05} shows the transition point $t_*$ based upon the variation of $p_d$ determined from simulations with the same basic settings as in Figure~\ref{fig:04}. We see in Figure~\ref{fig:05} that over a wide range of parameters, the transition $t_*$ is insensitive to the probability $p_d$. Only if $p_d$ is very high around $p_d=0.5$, the value of $t_*$ rises. From a medical standpoint, this is a good result as it means we can focus on increasing the external control parameters but do not really have to worry much about the death of components in the initial microscopic phase.

\section{Summary and Outlook}
\label{sec:summary}

Despite their quite extreme simplicity, our two models for cancer dynamics and organism response did provide surprisingly broad insights into the challenges of enhancing/using self-adaptation. Of course, another argument would be that we should not be surprised by this as for complex systems such as epidemic dynamics, climate dynamics or opinion formation very simple mathematical models have already provided remarkable insights. It seems that this shift of perspective could also be useful to combat cancer. Indeed, we basically have no hope to build a faithful full organism model for humans as the systems biology is far too complicated. Hence, the best we can hope for are reduced partial models. From these models, it might be possible to understand the most important features of large-scale whole body dynamics, which might be very relevant to treat cancer, or even many other dangerous medical conditions. Investigating this is far beyond the initial mathematical modeling we have discussed in this work. However, it provides an interesting shift of perspective as instead of specializing and trying to isolate subsystems of an organism, we allow instead for more coupling/interaction between subsystems at the expense fine-grained model details. 

The key technical results we have obtained are the crucial time scales between adaptation and cancer progression in our models. In particular, we have identified in the macroscopic model \textcolor{black}{bifurcations} that led to a cancer-free steady state. Although this state is attracting for many initial conditions, we have also shown that large transient spikes can occur in cancer cell density, which requires further speed-up for adaptation. We have studied this speed-up using a microscopic probabilistic model and identified a \textcolor{black}{quite sharp and localized} transition to successful conversion to components that combat cancer. Yet, this transition does shift its location based on parameters and could potentially be difficult to find and control in practice. \textcolor{black}{Yet, coming back to our initial biological motivation, we have at least suggested one plausible explanation using our two models, why certain cancer therapies that aim to exploit body adaptation mechanisms currently may fail frequently: the time scale to the transition point $t_*$ could potentially be very long so that initial therapy shows no effect initially or even on some longer time scales. Therefore, one might abandon such approaches in practice, i.e., the potentially helpful bifurcations and transitions are only visible on longer time scales.}  

\textcolor{black}{Several technical steps could certainly now be carried out from a mathematical perspective. A much more detailed bifurcation analysis of the \textit{brisc} ODE model is certainly possible, e.g., trying evem harder to find bifurcations of higher co-dimension as organizing centers for more complex dynamics. Although Figure~\ref{fig:06}(b) indicates that higher-codimension bifurcations might not be particularly interesting for the main biological mechanisms in our model. Also a more global convergence-speed analysis to steady state could be a reachable goal. Even more ambitiously, one could aim for a precise description of finite-time dynamics and transients, e.g., using finite-time Lyapunov exponents (FTLEs)~\cite{DoanKarraschNguyenSiegmund,LaiTel}. Interestingly, although classical Lyapunov exponents are very commonly used in cancer modelling~\cite{GallasGallasGallas,KhajanchiPercGhosh,PosadasCrileyCoffey}, FTLEs are only computed very rarely at this point for cancer dynamics~\cite{Leeetal}. We have shown that the finite time scale is crucial, so FTLEs are a natural direction. Of course, there is always the possibility to extend the models. It is well-understood that for biological systems spatial structure, stochastic noise, communication delays, heterogeneous coupling, and many other additional modelling aspects can be highly relevant. In this work, we have tried to keep these structures at the minimum to extract the main dynamical features first.} Similar remarks apply to investigating generalizations \textcolor{black}{of the reaction terms as discussed in the more general \textit{risc} ODE models.} For the microscopic probabilistic model, it seems plausible that it could eventually be analyzed using classical \textcolor{black}{probabilistic techniques}. It might then also be possible to study the \textcolor{black}{transition at $t_*$} in more detail analytically, e.g., regarding the existence of invariant measures\textcolor{black}{ and scaling laws for the transition point. This rigorous mathematical analysis of the microscopic model is currently in progress but will be published separately as it turns out to be surprisingly lengthy and mathematically involved. In summary, although further analysis steps are definitely worth pursuing, we are going to postpone the microscopic model analysis to keep the main modeling ideas and results in a succinct form.}


\appendix

\subsection*{Conflict of Interest}

The author declares that he has no conflicts of interest. 

\subsection*{Data Availability}

Not applicable as no data was used in this study.

\end{document}